\documentclass[11pt]{article}
\usepackage{amsmath,amssymb,amsfonts,amsthm,latexsym,epsfig,hyperref,array,enumerate,geometry,graphicx,url,tikz,float,lipsum}
\usepackage[all]{xy}
\newtheorem{theorem}{\bf Theorem}
\newtheorem{lemma}[theorem]{\bf Lemma}
\newtheorem{proposition}[theorem]{\bf Proposition}
\newtheorem{corollary}[theorem]{\bf Corollary}

\newtheorem{remark}[theorem]{\bf Remark}
\newtheorem{definition}[theorem]{\bf Definition}
\newtheorem{example}[theorem]{\bf Example}
\usepackage[labelfont=bf]{caption}

\usepackage [english]{babel}
\usepackage [autostyle, english = american]{csquotes}
\MakeOuterQuote{"}

\begin{document}
	\title{\bf A Note on Vectorial Boolean Functions as Embeddings}
	
	\author{Augustine Musukwa$^1$ and Massimiliano Sala$^2$}
	%If there are more than two authors, 'et al.' is used.
	\date{}
	\maketitle  
	
	\begin{center}{$^1$Mzuzu University, P/Bag 201, Mzuzu 2, Malawi\\ $^{1,2}$University of Trento, Via Sommarive, 14, 38123 Povo, Trento, Italy \\$^1$augustinemusukwa@gmail.com\\ $^2$maxsalacodes@gmail.com } \end{center}            % typeset the header of the contribution
	\begin{abstract}
		\noindent Let \( F \) be a vectorial Boolean function from \( \mathbb{F}_2^n \) to \( \mathbb{F}_2^m \), with \( m \geq  n \). We define \( F \) as an embedding if \( F \) is injective. In this paper, we examine the component functions of \( F \), focusing on constant and balanced components. Our findings reveal that at most \( 2^m - 2^{m-n} \) components of \( F \) can be balanced, and this maximum is achieved precisely when \( F \) is an embedding, with the remaining \( 2^{m-n} \) components being constants. Additionally, for partially-bent embeddings, we demonstrate that there are always at least \( 2^n - 1 \) balanced components when \( n \) is even, and \( 2^{m-1} + 2^{n-1} - 1 \) balanced components when \( n \) is odd. A relation with APN functions is shown.
	\end{abstract}
	
	
	\section{Introduction}
	Vectorial Boolean functions play a fundamental role in cryptography, particularly in the design of block ciphers and stream ciphers \cite{Car1}. They are used as S-boxes (substitution boxes) to provide non-linearity and resistance against linear and differential cryptanalysis. Strong cryptographic properties such as high nonlinearity, differential uniformity, algebraic immunity, and correlation immunity make these functions essential for securing encryption algorithms.
	
	In this note, we study vectorial Boolean functions mapping from a lower-dimensional vector space to a higher-dimensional vector space. Although few papers have been published on these functions, some studies have suggested their possible application in ciphers. Recently, some researchers have considered these functions, with a focus on studying APN functions that satisfy the so-called D-property (see, for example, \cite{Abbo,Tani}). These functions were also partially studied in \cite{Arag}, where the authors examined certain cryptographic properties of Feistel networks using functions from $\mathbb{F}_2^n$ to $\mathbb{F}_2^m$ as S-boxes. Specifically, they provided an example of a cipher using an injective APN function from 4 bits to 5 bits as an S-box. This suggests further investigation of injective maps, which in this paper we call "embeddings." 
	
	In our study of these embeddings, we are also motivated by two more reasons. The first is of a geometric type. If we consider an affine map $F$ between two vector spaces $\mathbb{F}_2^n$ and $\mathbb{F}_2^m$, with $m\geq n$, $F$ will have the property that all components are either balanced or constant (we call such an $F$ "affine-like" in this introduction). There is no hope that generic vectorial Boolean functions share this property, but it is well-known that bijections do, since all their components are balanced except for the trivial zero component. Along the same reasoning, we try to understand when a function has many balanced components (which may be seen as a measure of "balancedness"). Since bijections have a huge number of balanced components, we attempt to identify a more general property that guarantees high balancedness.
	We do that by looking at embeddings and indeed we present in this paper some results along this line. First, whatever the degree of the embedding $F$, if the image happens to be a subspace then $F$ is affine-like (Theorem \ref{balanced-ImageSpace}). Second, if $F$ is quadratic (or partially-bent), we can show the presence of a huge number of balanced components (especially when $n$ is odd), even if others can clearly be anything (Theorem~\ref{least-balanced-component}). Finally, coming back to embeddings of any degree, 
	$F$ is affine-like whenever
	$F$ has the highest possible number (permitted by the space dimensions $m,n$)
	of balanced components
	(Corollary \ref{balanced-components}) or of constant components
	(Corollary \ref{constant-components}).
	
	The last motivation is related to the theory of APN functions. Although cubic APN permutations in odd dimension are known, none is known in even dimension. Certainly, there is none for $n=4$ (as experimentally shown in \cite{Hou} and formally proved in \cite{Cal})
	and for $n=6$ (as experimentally shown in \cite{Lan}).
	The reason for this scarcity is still unclear. For APN functions, suitable restrictions of their derivatives are embeddings. We show in Section \ref{appl} that, for such restrictions, at least half of their components must be balanced in the odd case, while in the even case a surprising three quarters of them must be balanced.
	This significant difference
	might help explain the scarcity of their permutation classes.
	
	This paper consists of six sections and is structured as follows. Section \ref{preliminaries} provides an overview of some known results and essential notations for our study. In Section \ref{preliminary-results}, we present some preliminary results, with a primary focus on the weight of first-order derivatives of any Boolean function. Section \ref{main-results} presents a study of some quantities for vBfs from $\mathbb{F}_2^n$ to $\mathbb{F}_2^m$. In Section \ref{main-results-1}, we focus on the case $m\geq n$. Here, we investigate the balancedness of vBfs, with a particular emphasis on injective functions, and we identify cases when they are affine-like. Finally,  an application to cubic APN functions is given in Section \ref{appl}.

	\section{Preliminaries}\label{preliminaries}
	In this section, we provide some definitions, known results, and notations essential to our work. For further details, we recommend referring to \cite{Car1,Car2,Mac,Mus,Mus1}.
	
	The field of two elements, \(0\) and \(1\), is denoted by \(\mathbb{F}_2\). Let $n$ and $m$ be positive integers. The vector space of dimension \(n\) over \(\mathbb{F}_2\) is represented as \(\mathbb{F}_2^n\). We use $0_n$ to denote the zero vector $(0,\ldots, 0)$ in $\mathbb{F}_2^n$. The size of any finite set \(A\) is denoted by\(|A|\).
	
	A function \(F\) from \(\mathbb{F}_2^n\) to \(\mathbb{F}_2^m\) is referred to as a {\em vectorial Boolean function (vBf)}. We simply say a {\em Boolean function} when \(m=1\) and it is represented with lowercase letters, such as \(f\). The set of all Boolean functions on $n$ variables is denoted by $B_n$. 
	
	A Boolean function on $n$ variables in Algebraic Normal Form (ANF) is given by: \[f(x_1,...,x_n)=\sum_{I\subseteq \mathcal{P}}a_I\left(\prod_{i\in I}x_i\right),\] where $\mathcal{P}=\{1,...,n\}$ and $a_I\in \mathbb{F}_2$ for all $I\subseteq P$. The {\em algebraic degree} (or simply {\em degree}) of $f$, denoted by $\deg(f)$, is $\max_{I\subseteq \mathcal{P}}\{|I| : a_I\ne 0\}$.
	
	The {\em weight} of $f\in B_n$ is defined as $\mathrm{wt}(f)=|\{x\in \mathbb{F}_2^n : f(x)=1\}|$. We say that $f$ is {\em balanced} if $|\{x\in \mathbb{F}_2^n: f(x)=1\}|=|\{x\in \mathbb{F}_2^n : f(x)=0\}|$, that is, if $\mathrm{wt}(f)=2^{n-1}$. The {\em evaluation} of \(f\in B_n\) is defined by \({\rm ev}(f) = (f(x))_{x \in \mathbb{F}_2^n}\) (with an implicit ordering on $\mathbb{F}_2^n$).
	
	A Boolean function $f\in B_n$ is called {\em constant} if $\deg(f)=0$, {\em affine} if $\deg(f)\leq 1$, {\em linear} if $\deg(f)=1$ and $f(0_n)=0$, {\em quadratic} if $\deg(f)\le 2$ and {\em strictly quadratic} if $\deg(f)=2$. Let $\alpha\in\mathbb{F}_2^n$ and $c\in\mathbb{F}_2$. An affine function is given by $\alpha\cdot x+c$. Let $\alpha\ne 0_n$, a linear function is given by $\alpha\cdot x$. Any affine function with $\alpha \neq 0_n$ is balanced \cite{Car1}. 
	
	The {\em Walsh transform} of \(f \in B_n\) is defined as the function \(\mathcal{W}_f: \mathbb{F}_2^n \to \mathbb{Z}\) given by \(\mathcal{W}_f(a) = \sum_{x \in \mathbb{F}_2^n} (-1)^{f(x) + a \cdot x},\) for any \(a \in \mathbb{F}_2^n\). We denote \(\mathcal{F}(f)=\mathcal{W}_f(0_n)\).

	\begin{remark}\label{fourier} \cite [p.76]{Car1}
		For any $f\in B_n$, \(\mathcal{F}(f)=2^n-2\mathrm{wt}(f).\)  \\Therefore, \(f\) is balanced if and only if \(\mathcal{F}(f) = 0\).
	\end{remark}
	
	The {\em nonlinearity} of $f\in B_n$  is given by $\mathcal{N}(f)=2^{n-1}-\frac{1}{2}\max_{a\in\mathbb{F}_2^n}|\mathcal{W}_f(a)|$. It is well-known that $\mathcal{N}(f)\leq 2^{n-1}-2^{\frac{n}{2}-1}$ \cite[p.80]{Car1}. A Boolean function $f\in B_n$ is called {\em bent} if $\mathcal{N}(f)=2^{n-1}-2^{\frac{n}{2}-1}$ and this may happen only when $n$ is even. A Boolean function $f\in B_n$ is called {\em semi-bent} if $\mathcal{N}(f)=2^{n-1}-2^{\frac{n-1}{2}}$ and this many happen only when $n$ is odd.
	
	We define the {\em first-order derivative $D_af$} of $f\in B_n$ in the direction of $a\in\mathbb{F}_2^n$ by\\ $D_af(x)=f(x+a)+f(x)$. A Boolean function $f\in B_n$ is bent if and only if  $D_af$ is balanced for all nonzero $a\in \mathbb{F}_2^n$ \cite[p.82]{Car1}.  
	
	Any $a\in\mathbb{F}_2^n$ is called a {\em linear structure} of $f$ if $D_af$ is constant. The set of all linear structures of \( f \), denoted by \( V(f) \), is a subspace of \( \mathbb{F}_2^n\) \cite[p.99]{Car1}. A Boolean function $f$ is {\em partially-bent} if there exists a linear subspace $E$ of $\mathbb{F}_2^n$ such that the restriction of $f$ to $E$ is affine and the restriction of $f$ to any complementary subspace $E'$ of $E$ (where $E\oplus E'=\mathbb{F}_2^n$) is bent. This also includes the case $E=\mathbb{F}_2^n$, i.e., $\deg(f)\leq 1$. It was observed in \cite{Cal} that $E=V(f)$ for a partially-bent function $f$. It is well-known that every quadratic function is partially-bent \cite[p.257]{Car1}.
	
	As it can be noted in \cite[p.257]{Car1} and \cite[p.141]{Car3}, the quantity \(\mathcal{F}(f)\) of any unbalanced partially-bent function \(f\) is linked to \(\dim V(f)\). We present this result in the following. 
	
	\begin{remark}\label{fourier-partially-bent}
		For any unbalanced partially-bent function $f\in B_n$, let $k=\dim V(f)$. Then $k\equiv n\pmod{2}$ \ and \ $\mathcal{F}(f)\in
		\Big\lbrace\pm2^{\frac{n+k}{2}} \Big\rbrace$. \\Therefore, for $n$ even, we have that $\dim V(f)=0$ if and only if f is bent,\\ while for $n$ odd $\dim V(f)\ge 1$, $\dim V(f)=1$ if and only if f is semi-bent.
	\end{remark}

	%%%%%%%%%%%%%%%%%%%%%%%%%%%%%%%%%%%%%%%%%%%%%%%%%%%%%%%%%%%%%%%%%%%%%%%%%%%%%%%%%%%%%%%%%%%%%%%%%%%%%%%%%%%%%%%%%%%%%%%%%%%%%%%%%%%%%%%%%%%%%%%%%%%%%
	A vBf $F$ from $\mathbb{F}_2^n$ to $\mathbb{F}_2^m$ can be represented by $F=(f_1,...,f_m)$, where $f_1,...,f_m$ are Boolean functions on $n$ variables called {\em coordinate functions}. The Boolean functions given by $\lambda\cdot F$, with $\lambda\in \mathbb{F}_2^m$ and "$\cdot$" denoting the dot product, are called {\em component functions} (or {\em components}) of $F$ and they are denoted by $F_\lambda$. The components $\{F_\lambda\}_{\lambda\ne 0_m}$ are called {\em nontrivial components}. The degree of $F$ is given by $\deg(F)=\max_{\lambda\in \mathbb{F}^m}\deg(F_\lambda)$, which clearly is the maximum degree over the coordinate functions. 
	
	The first-order derivative of a vBf $F$ in the direction of any $a\in\mathbb{F}_2^n$ is defined in a similar manner as for Boolean functions, that is, it is given by  $D_aF(x)=F(x+a)+F(x)$.
	
	Let $X$ and $Y$ be finite sets. Let $G$ be a function from $X$ to $Y$. The {\em pre-image} of $y\in Y$ by $G$ is defined by $G^{-1}(y)=\{x\in X : G(x)=y\}$. We say that \(G\) is {\em injective} if $x_1\ne x_2$ implies $G(x_1)\ne G(x_2)$ for all $x_1,x_2\in X$. Equivalently, we say $G$ is injective if $|G^{-1}(y)|\leq 1$ for all $y\in Y$. The {\em image} of a function \(G\) is defined as \({\rm Im}(G) = \{G(x) : x \in X\}\subseteq Y\). When $F$ is a vBf observe that the definition of injectivity is equivalent to stating that \(|{\rm Im}(F)| = 2^n\), which turns out to be very useful in some proofs in this paper. We call an injective vBf an {\em embedding}.

	A vBf $F$ is {\em balanced} if $|F^{-1}(y)|=2^{n-m}$ for all $y\in\mathbb{F}_2^m$. For  $n=m$, a balanced vBf is a permutation, which clearly is an embedding. A vBf $F$ is balanced if and only if all nontrivial components are balanced \cite[p.112]{Car1}. Let \(B(F)=\{\lambda\in\mathbb{F}_2^m : F_\lambda \text{ is balanced}\}\). Since $F$ is balanced if and only if $|B(F)|=2^m-1$, we can view $|B(F)|$ as a measure of {\em balancedness} for $F$.
	
	We call a vBf $F$ a {\em quadratic embedding} if it is an embedding with $\deg(F)=2$.  We say that $F$ is a {\em partially-bent embedding} if it is an embedding with $\deg(F)\ge 2$ and all nontrivial components are partially-bent.
	
	\begin{remark}
		According to our definitions, all quadratic embeddings are partially-bent embeddings. Moreover, a quadratic/partially-bent embedding $F$ may have affine nontrivial components, but $F$ must also have higher-degree components (with degree two if $F$ is quadratic).
	\end{remark}
	
	If $F$ is a vBf from $\mathbb{F}_2^n$ to $\mathbb{F}_2^m$, we define $C(F)=\{\lambda\in\mathbb{F}_2^m : F_\lambda \text{ is constant}\}$. Clearly, $C(F)$ is a vector space. A function $\psi_n$ from $\mathbb{F}_2^n$ to itself is called an {\em affinity} if there is an invertible matrix $M$ and a vector $t\in \mathbb{F}_2^n$ such that $\psi_n(x)=xM+t$, for any $x\in \mathbb{F}_2^n$ (using row vector notation).
	Any two vBfs $F$ and $F'$ are called {\em affine equivalent} if there are two affinities, $\psi_n$ and $\psi_m$, such that $F'=\psi_m \circ F \circ \psi_n$ \cite{Car1}. Two affine equivalent vBfs share many properties, but for our goals it is enough to know that $|B(F)|=|B(F')|$ and $|C(F)|=|C(F')|$.
	%%%%%%%%%%%%%%%%%%%%%%%%%%%%%%%%%%%%%%%%%%%%%%%%%%%%%%%%%%%%%%%%%%%%%%%%%%%%%%%%%%%%%%%%%%%%%%%%%%%%%%%%%%%%%%%%%%%%%%%%%%%%%%%%%%%%%%%%%%%%%%%%%%
	
	\section{Preliminary results on Boolean functions}\label{preliminary-results}
	
	In this section, we focus on Boolean functions and determine the sum of the weights of their first-order derivatives. We are specifically interested in the result for quadratic and partially-bent functions because it is used in the next section. We begin with the following result, which can also be found in \cite[p.1496]{Car2}. We report this proof to accustom the reader to our proving methods.
	
	\begin{lemma}\label{square-fourier}
		For any $f\in B_n$, we have
		\begin{align*}
			\mathcal{F}^2(f)=\sum_{a\in\mathbb{F}_2^n}\mathcal{F}(D_af).
		\end{align*}
	\end{lemma}
	
	\begin{proof}
		We have the following:
		\begin{align*}
			\mathcal{F}^2(f)&=\left(\sum_{x\in\mathbb{F}_2^n}(-1)^{f(x)}\right)^2=\sum_{x,y\in\mathbb{F}_2^n}(-1)^{f(x)+f(y)}= \sum_{x,a\in\mathbb{F}_2^n}(-1)^{f(x)+f(x+a)}\\&=\sum_{x,a\in\mathbb{F}_2^n}(-1)^{D_af(x)}=\sum_{a\in\mathbb{F}_2^n}\sum_{x\in\mathbb{F}_2^n}(-1)^{D_af(x)}= \sum_{a\in\mathbb{F}_2^n}\mathcal{F}(D_af).\qedhere
		\end{align*}
	\end{proof}
	
	By applying  Lemma  \ref{square-fourier} and Remark \ref{fourier}, we can easily infer the following. 
	
	\begin{corollary} \label{weight-derivatives-square-fourier}
		Let $f\in B_n$. Then
		\[\sum_{a\in\mathbb{F}_2^n}{\rm wt}(D_af)=2^{2n-1}-\frac{1}{2}\mathcal{F}^2(f).\]
	\end{corollary}
	
	%\begin{proof}
	%	By Remark \ref{fourier} and Lemma \ref{square-fourier}, we can obtain the sum of weights of all the first-order derivatives of $f$ in terms of $\mathcal{F}^2(f)$ as follows:
	%	\begin{align*}
		%		\mathcal{F}^2(f)&=\sum_{a\in\mathbb{F}_2^n}\mathcal{F}(D_af)=\sum_{a\in\mathbb{F}_2^n}\left(2^n-2{\rm wt}(D_af)\right)=2^{2n}-2\sum_{a\in\mathbb{F}_2^n}{\rm wt}(D_af)
		%	\end{align*} from which we deduce the claimed equality.
	%\end{proof}
	
	From Corollary \ref{weight-derivatives-square-fourier}, Remark \ref{fourier} and the inequality $\mathcal{F}^2(f)\ge 0$, we immediately get the following lemma.
	
	\begin{lemma}\label{sum-weight-derivatives-balanced}
		For any $f\in B_n$, we have
		\[\sum_{a\in\mathbb{F}_2^n\setminus\{0_n\}}{\rm wt}(D_af)=\sum_{a\in\mathbb{F}_2^n}{\rm wt}(D_af)\leq 2^{2n-1}.\] Furthermore, equality holds if and only if $f$ is balanced.
	\end{lemma}
	
	We next determine the value \ $\sum_{a\in\mathbb{F}_2^n\setminus\{0_n\}}{\rm wt}(D_af)$ \ for partially-bent functions.

	\begin{corollary}\label{sum-weight-partially-bent-derivatives}
		Let $f\in B_n$ be partially-bent and let $k=\dim V(f)$. Then 
		\begin{itemize}
			\item[(i)] \(\sum_{a\in\mathbb{F}_2^n\setminus\{0_n\}}{\rm wt}(D_af)= 2^{2n-1}\) if $f$ is balanced,
			\item[(ii)] \( \sum_{a\in\mathbb{F}_2^n\setminus\{0_n\}}{\rm wt}(D_af)=2^{2n-1}-2^{n+k-1}\) if $f$ is unbalanced.
		\end{itemize}
	\end{corollary}
	\begin{proof}
		If $f$ is balanced, we just apply Lemma \ref{sum-weight-derivatives-balanced}.  \\Otherwise, we just apply Remark \ref{fourier-partially-bent} to Corollary \ref{weight-derivatives-square-fourier}.
	\end{proof}

	\section{Results for general vectorial Boolean functions}\label{main-results}
	In this section, we study some quantities for vBfs from $\mathbb{F}_2^n$ to $\mathbb{F}_2^m$, where $n$ and $m$ are any positive integers. We begin with the following lemma.

	\begin{lemma} \label{nonpermutation-eqn}
		Let $F$ be a vBf from $\mathbb{F}_2^n$ to $\mathbb{F}_2^m$. Let $a\in\mathbb{F}_2^n$. Then 
		\begin{align}\label{nonpermutation-eqn-2}
			\sum_{x\in\mathbb{F}_2^n}\sum_{\lambda\in\mathbb{F}_2^m}(-1)^{\lambda\cdot (F(x)+F(x+a))}=\begin{cases}
				2^{n+m} & \text{ if } a=0_n\\
				2^m\,|\{x \in \mathbb{F}_2^n : F(x)=F(x+a)\}| & \text{ if } a\ne 0_n.
			\end{cases}
		\end{align}
	\end{lemma}
	
	\begin{proof}
		Observe that $F(x)=F(x+a)$ implies that $F(x)+F(x+a)=0_m$ and in this case \[\sum_{\lambda\in\mathbb{F}_2^m}(-1)^{\lambda\cdot (F(x)+F(x+a))}=\sum_{\lambda\in\mathbb{F}_2^m}(-1)^{\lambda\cdot 0_m}=2^m.\] When $a=0_n$, we have a special case of the above equation, since $F(x)+F(x+a)=F(x)+F(x)=0_m$.  
		
		On the other hand,
		$F(x)\neq F(x+a)$ implies that $F(x)+F(x+a)=\varepsilon\neq 0_m$ and so  \begin{align*}\label{nonpermutation-eqn-1} \sum_{\lambda\in\mathbb{F}_2^m}(-1)^{\lambda\cdot (F(x)+F(x+a))}=\sum_{\lambda\in\mathbb{F}_2^m}(-1)^{\lambda\cdot \varepsilon}=0.\end{align*} We get $0$ in the preceding equation because $\lambda\mapsto\lambda\cdot \varepsilon$ is a linear Boolean function in $\lambda$.
		
		By summing over $x$ the previous equalities, we obtain Equation (\ref{nonpermutation-eqn-2}). 
	\end{proof}

	Next we establish a relation that characterizes the embeddings and it is useful in determining the number of balanced components in the next section.

	\begin{theorem}\label{square-fourier-vectorial}
		Let $F$ be a vBf from $\mathbb{F}_2^n$ to $\mathbb{F}_2^m$. Then
		\[\sum_{\lambda\in\mathbb{F}_2^{m}}\mathcal{F}^2(F_\lambda)\geq 2^{n+m}.\]
		Moreover, equality holds if and only if \ $m\ge n$ and $F$ is an embedding.
	\end{theorem}
	
	\begin{proof}
		Let $\varphi_F(a)=\sum_{x\in\mathbb{F}_2^n}\sum_{\lambda\in\mathbb{F}_2^m}(-1)^{\lambda\cdot (F(x)+F(x+a))}$. From Lemma \ref{nonpermutation-eqn}, it is obvious that \begin{itemize}
			\item[i)] $\varphi_F(a)\ge 0$, for any $a\in\mathbb{F}_2^n$,
			\item[ii)] $\varphi_F(a)=0$ if and only if $F(x)\ne F(x+a)$ for all $x\in\mathbb{F}_2^n$.
		\end{itemize}
		
		Note that $F$ is an embedding if and only if, for any $a\ne 0_n$, we have $F(x)+F(x+a)\neq 0_m$ for all $x\in\mathbb{F}_2^n$. Hence, $F$ is embedding if and only if $\varphi_F(a)=0$ for any $a\ne 0_n$. 
		
		Now we look at the quantity $\sum_{a\ne 0_n}\varphi_F(a)$. \\For any $F$, \begin{align}\label{eqn-4}
			\sum_{a\ne 0_n}\varphi_F(a)=\sum_{a\ne 0_n}\sum_{x\in\mathbb{F}_2^n}\sum_{\lambda\in\mathbb{F}_2^m}(-1)^{\lambda\cdot (F(x)+F(x+a))}\ge 0,
		\end{align} while $F$ is an embedding if and only if \ $\sum_{a\ne 0_n}\varphi_F(a)=0$.
		
		Next, we apply Relation (\ref{eqn-4}) and Lemma \ref{square-fourier} in the following:
		\begin{align}\label{eqn-5}
			\sum_{\lambda\in\mathbb{F}_2^m}\mathcal{F}^2(F_\lambda)&= \sum_{\lambda\in\mathbb{F}_2^m}\left(\sum_{a\in\mathbb{F}_2^n}\mathcal{F}(D_aF_\lambda)\right)= \sum_{\lambda\in\mathbb{F}_2^m}\sum_{a\in\mathbb{F}_2^n}\sum_{x\in\mathbb{F}_2^n}(-1)^{D_aF_\lambda(x)}\nonumber\\&=\sum_{a\in\mathbb{F}_2^n}\sum_{x\in\mathbb{F}_2^n}\sum_{\lambda\in\mathbb{F}_2^m}(-1)^{\lambda\cdot \left(F(x)+F(x+a)\right)}\nonumber\\&= \sum_{x\in\mathbb{F}_2^n}\sum_{\lambda\in\mathbb{F}_2^m}(-1)^{\lambda\cdot \left(F(x)+F(x+0_n)\right)}+\sum_{a\ne 0_n}\sum_{x\in\mathbb{F}_2^n}\sum_{\lambda\in\mathbb{F}_2^m}(-1)^{\lambda\cdot \left(F(x)+F(x+a)\right)}\nonumber\\&=\sum_{x\in\mathbb{F}_2^n}\sum_{\lambda\in\mathbb{F}_2^m}(-1)^{\lambda\cdot 0_m}+\sum_{a\ne 0_n}\varphi_F(a)=2^{n+m} +\sum_{a\ne 0_n}\varphi_F(a)\nonumber\\&\geq 2^{n+m}.
		\end{align}
		In Relation (\ref{eqn-5}), equality holds if and only if $F$ is an embedding.
	\end{proof}
	\begin{remark}
		Worth mentioning is the study in \cite[p.126]{Car4}, where the value $\sum_{\lambda\in\mathbb{F}_2^m}\mathcal{F}^2(F_\lambda)$ was considered and shown to be \[\sum_{\lambda\in\mathbb{F}_2^m}\mathcal{F}^2(F_\lambda)=2^m\sum_{b\in\mathbb{F}_2^m}|F^{-1}(b)|^2.\] Theorem \ref{square-fourier-vectorial} could also be proved by using this relation.
	\end{remark}
	
	We next apply Lemma \ref{square-fourier} and Theorem \ref{square-fourier-vectorial} to directly deduce the following.
	
	\begin{corollary}\label{fourier-derivatives-vectorial}
		Let $F$ be a vBf from $\mathbb{F}_2^n$ to $\mathbb{F}_2^m$. Then we have 
		\[\sum_{\lambda\in\mathbb{F}_2^{m}}\sum_{a\in\mathbb{F}_2^n}\mathcal{F}(D_aF_\lambda)\geq 2^{n+m}.\]
		Moreover, equality holds if and only if $m\geq n$ and  $F$ is an embedding.
	\end{corollary}
	
	\section{Results for vectorial Boolean functions with \texorpdfstring{$m\ge n$}{m >= n}}\label{main-results-1}
	This section explores vBf's from $\mathbb{F}_2^n$ to $\mathbb{F}_2^m$, with $m \ge n$. We investigate the balancedness of their component functions, with a particular emphasis on injective functions (embeddings). 
	
	Corollary \ref{fourier-derivatives-vectorial} carries significant implications for the weights of the first-order derivatives of components. This can be seen in the following.
	
	\begin{corollary}\label{sum-weight-derivatives-vectorial-Boolean-functions}
		Let $F$ be a vBf from $\mathbb{F}_2^n$ into $\mathbb{F}_2^m$, with $m\geq n$. Then \begin{align}\label{weight-embedding}\sum_{\lambda\in\mathbb{F}_2^m\setminus\{0_m\}}\sum_{a\in\mathbb{F}_2^n}{\rm wt}(D_aF_\lambda)\leq 2^{2n-1}(2^m-2^{m-n}).\end{align}  Moreover, equality holds if and only if $F$ is an embedding.
	\end{corollary}
	
	\begin{proof} We begin by observing that 
		\begin{align}\label{fourier-derivatives-zero-eqn}
			&\sum_{\lambda\in\mathbb{F}_2^{m}}\sum_{a\in\mathbb{F}_2^n}\mathcal{F}(D_aF_\lambda)=\sum_{a\in\mathbb{F}_2^n}\mathcal{F}(D_aF_{0_m})+\sum_{\lambda\in\mathbb{F}_2^{m}\setminus\{0_m\}}\sum_{a\in\mathbb{F}_2^n}\mathcal{F}(D_aF_\lambda)\nonumber\\&=2^{2n}+\sum_{\lambda\in\mathbb{F}_2^{m}\setminus\{0_m\}}\sum_{a\in\mathbb{F}_2^n}\mathcal{F}(D_aF_\lambda).
		\end{align}
		So, by Corollary \ref{fourier-derivatives-vectorial} and Equation (\ref{fourier-derivatives-zero-eqn}), we obtain the following:
		\begin{align}\label{fourier-to-weight-eqn}
			\sum_{\lambda\in\mathbb{F}_2^{m}\setminus\{0_m\}}\sum_{a\in\mathbb{F}_2^n}\mathcal{F}(D_aF_\lambda)\geq 2^{n+m}-2^{2n}=2^{2n}(2^{m-n}-1).
		\end{align} and equality holds if and only if $F$ is an embedding. \\ By Equation (\ref{fourier-to-weight-eqn}), we infer the following:
		
		\begin{align*}
			2^{2n}(2^{m-n}-1)& \leq \sum_{\lambda\in\mathbb{F}_2^m\setminus\{0_m\}}\sum_{a\in\mathbb{F}_2^n}\mathcal{F}(D_aF_\lambda)\\&=\sum_{\lambda\in\mathbb{F}_2^m\setminus\{0_m\}}\sum_{a\in\mathbb{F}_2^n}\left(2^n-2{\rm wt}(D_aF_\lambda)\right)\\&=\sum_{\lambda\in\mathbb{F}_2^m\setminus\{0_m\}}\left(2^{2n}-2\sum_{a\in\mathbb{F}_2^n}{\rm wt}(D_aF_\lambda)\right)\\&=2^{2n}(2^m-1)-2\sum_{\lambda\in\mathbb{F}_2^m\setminus\{0_m\}}\sum_{a\in\mathbb{F}_2^n}{\rm wt}(D_aF_\lambda)
		\end{align*} from which we obtain Relation (\ref{weight-embedding}) and equality holds if and only if $F$ is an embedding.	
	\end{proof}
	
	Next, we apply Corollary \ref{sum-weight-derivatives-vectorial-Boolean-functions} to determine an upper bound on the number of balanced components for {\em any} vBf with $m\ge n$. 
	
	\begin{corollary}\label{balanced-components}
		Let $F$ be a vBf from $\mathbb{F}_2^n$ to $\mathbb{F}_2^m$, with $m\geq n$. Then \(|B(F)|\le 2^m-2^{m-n}\). Furthermore, equality may be achieved only when $F$ is an embedding and in this case, all the other $2^{m-n}$components are constant.
	\end{corollary}
	
	\begin{proof}
		Let $F_\lambda$ be a component. By Lemma \ref{sum-weight-derivatives-balanced}, the value \(\sum_{a\in\mathbb{F}_2^n}{\rm wt}(D_aF_\lambda)\) is exactly $2^{2n-1}$ if $F_\lambda$ is balanced, strictly less otherwise. Considering (\ref{weight-embedding}), if $|B(F)|=2^m-2^{m-n}$ then the inequality reaches equality, if $|B(F)|>2^m-2^{m-n}$ then the inequality is contradicted.  Therefore, \(F\) can have at most \(2^m - 2^{m-n}\) balanced components. 
		
		When $|B(F)|=2^m-2^{m-n}$, $F$ must be an embedding since equality holds for (\ref{weight-embedding}). In this case, exactly \(2^{m-n}\) vectors \(\mu \in \mathbb{F}_2^m\) are such that \(\sum_{a \in \mathbb{F}_2^n} \mathrm{wt}(D_a F_\mu) = 0\), which occurs only when the corresponding components \(F_\mu\) are constants.
	\end{proof}
	
	\begin{remark}
		From Corollary \ref{balanced-components}, a vBf from \(\mathbb{F}_2^n\) to \(\mathbb{F}_2^m\), where \(m > n\), cannot have all its nontrivial components balanced, since it has at least $2^{m-n}$ unbalanced components.
	\end{remark}
	
	Let $F$ be a vBf from $\mathbb{F}_2^n$ to $\mathbb{F}_2^m$. Recall that $C(F)=\{\lambda\in\mathbb{F}_2^m : F_\lambda \text{ is constant}\}$ is a vector space. 
	
	\begin{lemma}\label{costants}
		Let $F$ be a vBf from $\mathbb{F}_2^n$ to $\mathbb{F}_2^m$ with $m\ge n$. Let $h=\dim C(F)$. Then
		$$
		|{\rm Im}(F)|\le 2^{m-h}\;.
		$$
	\end{lemma}
	
	\begin{proof}
		The case $h=0$ is obvious. 
		
		If $h\ge 1$, let $E$ be such that $\mathbb{F}_2^m=C(F)\oplus E$. We claim that 
		\begin{align}\label{surjection}
			\mbox{for any }x,y\in \mathbb{F}_2^n,\qquad  \pi_E\circ F(x)\ne \pi_E\circ F(y) \iff F(x) \ne F(y)\,,
		\end{align} where $\pi_E: \mathbb{F}_2^m\to E$ is any surjective function. From (\ref{surjection}) we have that $|{\rm Im}(\pi_E\circ F)|=|{\rm Im}(F)|$. Since ${\rm Im}(\pi_E\circ F)\subseteq E$ and $|E|=2^{\dim E}=2^{m-h}$, we have \[|{\rm Im}(F)|\le |E|=2^{m-h}.\] 
		
		To prove (\ref{surjection}), observe that ($\Rightarrow$) comes from $E \subset \mathbb{F}_2^m$, while ($\Leftarrow$) follows from the direct sum $\mathbb{F}_2^n=C(F)\oplus E$ and the equality (for any $x,y$) $\pi_{C(F)}\circ F(x)=\pi_{C(F)}\circ F(y)$.
	\end{proof}

	%Since $F_{\lambda_1+\lambda_2}=(\lambda_1+\lambda_2)\cdot F=\lambda_1\cdot F+\lambda_2\cdot F=F_{\lambda_1}+F_{\lambda_2}$, it is clear that $C(F)$ is a vector subspace of $\mathbb{F}_2^m$. 
	We again use Corollary \ref{sum-weight-derivatives-vectorial-Boolean-functions} to derive the following result on the number of constant components for an embedding.

	\begin{corollary}\label{constant-components}
		Let $F$ be an embedding from $\mathbb{F}_2^n$ to $\mathbb{F}_2^m$, with $m\geq n$. Then \(|C(F)|\le 2^{m-n}\). \\Furthermore, if equality is achieved, then  all the other $2^m-2^{m-n}$ components are balanced.
	\end{corollary}
	
	\begin{proof}
		According to Lemma \ref{costants}, $|{\rm Im}(F)|\le 2^{m-h}$, where $h=\dim(C(F))$. If $|C(F)|> 2^{m-n}$, then $h > m-n$ and so $|{\rm Im}(F)| < 2^{m-(m-n)}=2^n$, which is impossible.
		
		If \(|C(F)|=2^{m-n}\) and  $F$ is an embedding, then (\ref{weight-embedding}) becomes \begin{align}\label{8}\sum_{\lambda\in\mathbb{F}_2^m\setminus C(F)}\sum_{a\in\mathbb{F}_2^n}{\rm wt}(D_aF_\lambda)=2^{2n-1}(2^m-2^{m-n}).\end{align} But $|\mathbb{F}_2^m\setminus C(F)|=2^m-2^{m-n}$ and each addend in (\ref{8}) has to be at most $2^{2n-1}$, hence each addend is exactly $2^{2n-1}$. %So we conclude by Corollary \ref{sum-weight-partially-bent-derivatives} the $2^m-2^{m-n}$ components are balanced
	\end{proof}
	
	We next determine a lower bound on the number of balanced components for any partially-bent embedding. Observe that if $m=n$ then any embedding is a permutation, which implies that $|B(F)|=2^n-1$. So in the following statement, we are only interested in the case where $m$ is strictly larger than $n$.
	
	\begin{theorem}\label{least-balanced-component}
		Let $F$ be a partially-bent embedding from $\mathbb{F}_2^n$ to  $\mathbb{F}_2^m$, with $m\ge n+1$. Then 
		\begin{itemize}
			\item[(i)]  $|B(F)|\geq 2^n-1$, for $n$ even and equality holds if and only if all the other nontrivial components (i.e., the unbalanced components) are bent,
			\item[(ii)] $|B(F)|\geq 2^{m-1}+2^{n-1}-1$, for $n$ odd and equality holds if and only if all the other nontrivial components (i.e., the unbalanced components) are unbalanced semi-bent.
		\end{itemize}
	\end{theorem}
	
	\begin{proof}
		When $F_\lambda$ is not balanced, we write $k_\lambda=\dim V(F_\lambda)$. Observe from Remark \ref{fourier-partially-bent} that, for $n$ even, $k_\lambda$ is even with $0\leq k_\lambda\leq n$ and for $n$ odd, $k_\lambda$ is odd with $1\leq k_\lambda\leq n$. By Corollary \ref{sum-weight-partially-bent-derivatives}, we can derive the following equality:
		
		\begin{align}\label{least-quadratic-balanced-n-m-eqn-1}
			\sum_{\lambda\in\mathbb{F}_2^m\setminus\{0_m\}}&\sum_{a\in\mathbb{F}_2^n}{\rm wt}(D_aF_\lambda)= 2^{2n-1}|B(F)|+\sum_{\lambda\in\mathbb{F}_2^m\setminus(B(F)\cup\{0_m\})}\sum_{a\in\mathbb{F}_2^n}{\rm wt}(D_aF_\lambda)\nonumber\\&=2^{2n-1}|B(F)|+\sum_{\lambda\in\mathbb{F}_2^m\setminus(B(F)\cup\{0_m\})}(2^{2n-1}-2^{n+k_\lambda-1})\nonumber\\&=2^{2n-1}|B(F)|+2^{2n-1}(2^m-|B(F)|-1)-\sum_{\lambda\in\mathbb{F}_2^m\setminus(B(F)\cup\{0_m\})}2^{n+k_\lambda-1}\nonumber\\&= 2^{2n-1}(2^m-1)-2^{n-1}\sum_{\lambda\in\mathbb{F}_2^m\setminus(B(F)\cup\{0_m\})}2^{k_\lambda}
		\end{align} Since $F$ is an embedding, by Relation (\ref{least-quadratic-balanced-n-m-eqn-1}) and Corollary \ref{sum-weight-derivatives-vectorial-Boolean-functions}, we have 
		\[2^{2n-1}(2^m-1)-2^{n-1}\sum_{\lambda\in\mathbb{F}_2^m\setminus(B(F)\cup\{0_m\})}2^{k_\lambda} \quad\quad =\quad\quad 2^{2n-1}(2^m-2^{m-n}) \] from which we deduce the following:  
		\begin{align}\label{least-quadratic-balanced-n-m-eqn-2}
			&2^{n-1}\sum_{\lambda\in\mathbb{F}_2^m\setminus(B(F)\cup\{0_m\})}2^{k_\lambda} =2^{2n-1}(2^m-1)-2^{2n-1}(2^m-2^{m-n}) = 2^{2n-1}(2^{m-n}-1)\nonumber\\& \implies \sum_{\lambda\in\mathbb{F}_2^m\setminus(B(F)\cup\{0_m\})}2^{k_\lambda} = 2^n(2^{m-n}-1)\nonumber \\&\implies \sum_{\lambda\in\mathbb{F}_2^m\setminus(B(F)\cup\{0_m\})}2^{k_\lambda} =2^m-2^n.\end{align}
		Suppose that $n$ is even. We know that $k_\lambda\ge 0$. Thus,  \begin{align*} &2^m-2^n=\sum_{\lambda\in\mathbb{F}_2^m\setminus(B(F)\cup\{0_m\})}2^{k_\lambda}\geq \sum_{\lambda\in\mathbb{F}_2^m\setminus(B(F)\cup\{0_m\})}2^0 =2^m-|B(F)|-1\\& \implies |B(F)|\geq 2^n-1.\end{align*} Clearly, $|B(F)|=2^n-1$ if and only if $k_\lambda=0$ for all $\lambda\notin B(F)\cup \{0_m\}$ if and only if $F_\lambda$ is bent for all $\lambda\notin B(F)\cup \{0_m\}$.
		
		\noindent Suppose that $n$ odd, we know that $k_\lambda\geq 1$. Thus, \begin{align*}  &2^m-2^n=\sum_{\lambda\in\mathbb{F}_2^m\setminus(B(F)\cup\{0_m\})}2^{k_\lambda}\geq \sum_{\lambda\in\mathbb{F}_2^m\setminus(B(F)\cup\{0_m\})}2^1 =2(2^m-|B(F)|-1)\\& \implies |B(F)|\geq 2^{m-1}+2^{n-1}-1.\end{align*} $|B(F)|= 2^{m-1}+2^{n-1}-1$ if and only if $k_\lambda=1$ for all $\lambda\notin B(F)\cup \{0_m\}$ if and only if $F_\lambda$ is unbalanced semi-bent for all $\lambda\notin B(F)\cup \{0_m\}$.
	\end{proof}
	
	We next provide two examples of quadratic embeddings which meet the lower bounds of the number of balanced components. These examples were found with the help of MAGMA. The first example is an embedding from \(\mathbb{F}_2^3\) to \(\mathbb{F}_2^4\) and the second example is an embedding from \(\mathbb{F}_2^4\) to \(\mathbb{F}_2^5\).
	
	\begin{example}
		The following are the coordinate functions of a quadratic embedding $F$ from $\mathbb{F}_2^3$ to $\mathbb{F}_2^4$: $f_1=x_1x_2 + x_1 + x_2 + x_3, f_2=x_1x_3 + x_1 + x_2 + x_3, f_3=x_2x_3+x_1 + x_2, f_4=x_1 + x_3$. 
		\begin{align*}
			B(F)=\{&[ 0, 0, 0, 1 ],[ 0, 0, 1, 0 ], [ 0, 1, 0, 0 ], [ 0, 1, 0, 1 ], [ 0, 1, 1, 1 ], [ 1, 0, 0, 0 ],\\& [ 1, 0, 1, 0 ], [ 1, 0, 1, 1 ], [ 1, 1, 0, 1 ], [ 1, 1, 1, 0 ], [ 1, 1, 1, 1 ]\}.\end{align*} Observe that $n$ is odd and that $|B(F)|=2^3+2^2-1=11$, so the $15-11=4$ unbalanced components have to be semi-bent. Note that $14$ components are strictly quadratic and one component is linear.
	\end{example}

	\begin{example}\label{example-2}
		The following are the coordinate functions of a quadratic embedding $F$ from $\mathbb{F}_2^4$ to $\mathbb{F}_2^5$: $f_1=x_1x_2+x_4, f_2=x_1x_3+x_3+x_4, f_3=x_1x_4+x_3x_4+x_2, f_4=x_2x_3+x_3x_4+x_1+x_4, f_5=x_1x_3+x_2x_4$. 
		\begin{align*}
			B(F)=\{&[ 0, 0, 0, 1, 0 ],[ 0, 0, 1, 0, 0 ],[ 0, 0, 1, 1, 1 ],[ 0, 1, 0, 0, 0 ],[ 0, 1, 0, 0, 1 ],\\&[ 0, 1, 0, 1, 0 ], [ 0, 1, 0, 1, 1 ], [ 0, 1, 1, 0, 0 ], [ 0, 1, 1, 0, 1 ],[ 1, 0, 0, 0, 0 ],\\&[ 1, 0, 0, 1, 1 ],[ 1, 0, 1, 0, 1 ], [ 1, 0, 1, 1, 1 ], [ 1, 1, 0, 0, 0 ], [ 1, 1, 0, 0, 1 ]\}.
		\end{align*}
		Observe that $n$ is even and that $|B(F)|=2^4-1=15$, so the $31-15=16$ unbalanced components have to be bent. Note that all nontrivial components are strictly quadratic.
	\end{example}
	
	Since the two previous examples do not have nontrivial constant components, they cannot achieve the upper bound in Corollary \ref{balanced-components}.
	
	%\begin{remark}
	%	Let $F$ be an affine map, i.e., $\deg(F)=1$. The upper bound \\$B(F)=2^m-2^{m-n}$ from Corollary \ref{balanced-components} can be achieved also by $F$, but the lower bound $B(F)=2^n-1$ (even case) or $B(F)=2^{m-1}+2^{n-1}-1$ (odd case) cannot be reached by $F$.
	%\end{remark}

	\subsection{Special Case for embeddings}
	In this subsection, we consider two special cases, where we see that the bound of Corollary~\ref{balanced-components} can be actually tight.
	In the first case we restrict to $m=n+1$. The second case considers a situation where we see a behaviour in the embedding similar to that of an affine map. 
	
	We now examine the first case.

	\begin{proposition}\label{no-constants}
		Let a vBf $F$ from $\mathbb{F}_2^n$ to $\mathbb{F}_2^{n+1}$ be an embedding. Then
		\begin{itemize}
			\item[(i)] either there is one and only one constant nontrivial component and the other nontrivial components are balanced, or
			\item[(ii)] there is no constant nontrivial component.  
		\end{itemize}
	\end{proposition}
	
	\begin{proof}
		Since $F$ is embedding and $1\le |C(F)|\le 2^{n+1-n}=2$, then we only have two cases.
		
		If $|C(F)|=2$, we have one and only one constant nontrivial component, while the remaining nontrivial components must be balanced by Corollary \ref{constant-components}.
		
		If $|C(F)|=1$, $C(F)=\{0_{n+1}\}$ and $F$ does not have any constant nontrivial component.
	\end{proof}
	
	Finally, we examine a special case where the image of an embedding \(F\) from \(\mathbb{F}_2^n\) to \(\mathbb{F}_2^m\) forms a subspace of \(\mathbb{F}_2^m\). We present the result in the following.

	\begin{theorem}\label{balanced-ImageSpace}
		Let a vBf $F$ from $\mathbb{F}_2^n$ to $\mathbb{F}_2^m$ be an embedding, with $m\geq n$, and ${\rm Im}(F)$ be an affine subspace of $\mathbb{F}_2^m$. Then, 
		\begin{itemize}
			\item[(i)] for all $\lambda\in\mathbb{F}_2^m$, $F_\lambda$ is either constant or balanced. 
			\item[(ii)] $F$ has precisely $2^m-2^{m-n}$ balanced components and $2^{m-n}$ constant components.
		\end{itemize}
		
	\end{theorem}

	\begin{proof}
		Let $J=\{(v_1,\ldots, v_n,0,\ldots, 0) : v_i \in \mathbb{F}_2\}\subset \mathbb{F}_2^m$. By elementary linear algebra, ${\rm Im}(F)$ is an affine subspace of $\mathbb{F}_2^m$ if and only if there is an affinity $\psi: \mathbb{F}_2^m\to \mathbb{F}_2^m$ such that ${\rm Im}(F)=\psi(J)$. So if we consider $F'=\psi^{-1}\circ F$, $F': \mathbb{F}_2^n\to \mathbb{F}_2^m$, ${\rm Im}(F')=J$. Observe that we can choose $\psi$ so that $F'(x_1,\ldots,x_n)=(x_1,\ldots, x_n,0\ldots, 0)$. With this choice, $F'_\lambda$ is a linear combination of coordinate functions which are either linear or constant, so $|B(F')|+|C(F')|=2^m$, the same holds for $F_\lambda$ and we have (i).
		
		If $F'_\lambda$ is a linear combination of the last $m-n$ coordinate functions, then $F'_\lambda$ is constant and so $|C(F')|\ge 2^{m-n}$. By Corollary \ref{constant-components}, we have (ii).
	\end{proof}
	
	%%%%%%%%%%%%%%%%%%%%%%%%%%%%%%%%%%%%%%%%%%%%%%%%%%%%%%%%%%%%%%%%%%%%%%%%%%%%%%%%%%%%%%%%%%%%%%%%%%%%%%%%%%%%%%%%%%%%%%%%%%%%%%%%%%%%%%%%%%%%%%%%%%%%%%%%%%%%%%%%
	
	\section{Applications to cubic APN functions}\label{appl}
	Let $F$ be a vBF from $\mathbb{F}_2^n$ to itself. Let $I={(v_1,\ldots,v_{n-1},0)}$. We say that $F$ is an APN if for any $a\ne 0_n$ in $\mathbb{F}_2^n$, $x\ne y,y+a$, $F(x)+F(x+a)\not=F(y)+F(y+a)$ \cite[p.137]{Car1}.
	We can suppose w.l.o.g. that $a \not\in I$. Then $\mathbb{F}_2^n=I\oplus(I+a)$, and $D_aF(x)=D_aF(x+a)$. 
	\begin{definition}
		Let $F$ be a vBF from $\mathbb{F}_2^n$ to itself. Let $I_0={(v_1,\ldots,v_{n-1},0)}$ and $I_1={(v_1,\ldots,v_{n-1},1)}$. Let $a\in\mathbb{F}_2^n$, $a\ne 0_n$.  We define $\mathcal{D}_a(F): \mathbb{F}_2^{n-1}\to \mathbb{F}_2^n$ as $\mathcal{D}_a(F)(v)=D_aF(v,0)$ if $a\notin I_0$, as $\mathcal{D}_a(F)(v)=D_aF(v,1)$ if $a\notin I_1$.
	\end{definition} 
	Clearly, $F$ is APN if and only if $\mathcal{D}_a(F)$ is an embedding for any $a\ne 0_n$. Observe that if $F$ is cubic APN, then $\mathcal{D}_a(F)$ is quadratic/partially-bent embedding, for any $a\ne 0_n$.  Hence, by Theorem \ref{least-balanced-component}, the following is deduced.
	
	\begin{corollary}
		Let $F$ be a cubic APN function from $\mathbb{F}_2^n$ to itself. Then 
		\begin{itemize}
			\item[(i)]  $|B(\mathcal{D}_a(F))|\geq 2^{n-1}-1$, for $n$ odd and equality holds if and only if all the other nontrivial components (i.e., the unbalanced components) are bent,
			\item[(ii)] $|B(\mathcal{D}_a(F))|\geq 3\cdot 2^{n-2}-1$, for $n$ even and equality holds if and only if all the other nontrivial components (i.e., the unbalanced components) are unbalanced semi-bent.
		\end{itemize}
	\end{corollary}

\end{document}